\documentclass[twoside,a4paper]{article}

\usepackage{amsmath,graphicx,amssymb,fancyhdr,amsthm,enumerate,textcomp} 
\usepackage[usenames]{color}

\usepackage{hyperref}
\usepackage{float}
\usepackage{subcaption}
\usepackage[english]{babel}
\usepackage[numbers]{natbib}

\newcommand{\boldn}{\mbox{\boldmath$n$}} % bold k
\newcommand{\boldm}{\mbox{\boldmath$m$}} % bold m
\newtheorem{thm}{Theorem}[section]
\newtheorem{cor}[thm]{Corollary}
\newtheorem{lem}[thm]{Lemma}
\newtheorem{prop}[thm]{Proposition}

\theoremstyle{definition}
	
\theoremstyle{remark}

\usepackage{color}

\def\beq{\begin{eqnarray}}  
\def\eeq{\end{eqnarray}}  

\def\bsp{\begin{split}}  
\def\esp{\end{split}}

\def\d{\mathrm{d}}

\newcommand{\mbold}[1]{\mbox{\boldmath{\ensuremath{#1}}}}

\def \bl {\mbox{{\mbold\ell}}}
\def \bn {\mbox{{\bf n}}}
\def \bm {\mbox{{\bf m}}}

\def \bg {\mbox{{\mbold g}}}

\def \hn {\tilde{H}}
\def \wn{\tilde{W}}

\def\beq{\begin{eqnarray}}
\def\eeq{\end{eqnarray}}

\begin{document}

\title{\Large\textbf{Locally Boost Isotropic Spacetimes and the Type ${\bf D}^k$ Condition}}  
\author{{\large{D. McNutt${^1}$, A. Coley${^2}$, L. Wylleman${^1}$ and S. Hervik${^1}$ }    }
%EndAName    
%\address{    
\vspace{0.3cm} \\ 
{\small $^{1}$ Faculty of Science and Technology, University of Stavanger},\\ {\small  N-4036 Stavanger, Norway}\\
%\end{center}  \\
{\small $^{2}$ Department of Mathematics and Statistics, Dalhousie University}, \\ {\small Halifax, Nova Scotia, Canada B3H 3J5}\\
\vspace{0.3 mm} 
{\small E-mail: \texttt{david.d.mcnutt@uis.no, aac@mathstat.dal.ca,}} 
\\
\hspace{ 11 mm } {\small \texttt{ lode.wylleman@uis.no, sigbjorn.hervik@uis.no} }
}     
\date{}     
\maketitle   
\pagestyle{fancy}   
\fancyhead{} % clear all header fields   
\fancyhead[EC]{McNutt et al }   
\fancyhead[EL,OR]{\thepage}   
\fancyhead[OC]{Locally boost isotropic spacetimes and the type ${\bf D}^k$ condition}   
\fancyfoot{} % clear all footer fields   

\begin{abstract}

We consider the class of locally boost isotropic spacetimes in arbitrary dimension. For any spacetime with boost isotropy, the corresponding curvature tensor and all of its covariant derivatives must be simultaneously of alignment type {\bf D} relative to some common null frame. Such spacetimes are known as type ${\bf D}^k$ spacetimes and are contained within the subclass of degenerate Kundt spacetimes. Although, these spacetimes are $\mathcal{I}$-degenerate, it is possible to distinguish any two type ${\bf D}^k$ spacetimes, as the curvature tensor and its covariant derivatives can be characterized by the set of scalar polynomial curvature invariants for any type ${\bf D}^k$ spacetime. In this paper we find all type ${\bf D}^k$ spacetimes by identifying degenerate Kundt metrics that are of type ${\bf D}^k$ and determining the precise conditions on the metric functions.
\end{abstract}

\maketitle

\section{Introduction} \label{intro}

Any Lorentzian manifold, or spacetime, admitting a geodesic, shear-free, twist-free, non-expanding null congruence, $\bl$, belongs to the class of Kundt spacetimes \cite{exact}. The Kundt spacetimes are important solutions in four-dimensional (4D) general relativity (GR), with many applications for other field theories in 4D and in higher dimensions as well \cite{pravda2017exact, hervik2017universal, CFH2008}. For example, the near horizon geometry of the event horizon of a stationary asymptotically flat black hole spacetime is described by a Kundt spacetime \cite{MPS2018}. More generally, when the event horizon of a dynamical black hole does not interact with the exterior region, this is described by a non-expanding horizon (NEH), and the near horizon geometry of the NEH is again described by a Kundt spacetime \cite{Lewandowski2018}. In fact, these near horizon geometries belong to the subclass of {\it degenerate} Kundt spacetimes, to be defined below, which have the interesting property that they cannot be locally characterized by the set of their scalar polynomial curvature invariants ($SPIs$), $\mathcal{I}$, and they are said to be $\mathcal{I}$-degenerate.

To be precise, a spacetime is $\mathcal{I}$-degenerate if there exists a deformation of the metric $g_\tau$ such that $g_\tau$ is continuous in $\tau$, $g_0 = g$ and the limiting metric as $\tau \to \infty $, $g_\infty $ is not diffeomorphic to $g$ but the set of $SPIs$, $\mathcal{I}$, for $g$ and $g_\infty $ will be the same. Thus, the metric $g$ cannot be distinguished from $g_\tau$ using $SPIs$. A more practical definition of $\mathcal{I}$-degeneracy can be stated in terms of the structure of the curvature tensor and its covariant derivatives by exploiting the effect of a boost on a null coframe $\{ \bn, \bl, \bm^i\}$, $\bl' = \lambda \bl,~~ {\bf n'} = \lambda^{-1} \bn,$ for which the components of an
arbitrary tensor, ${\bf T}$, of rank $n$ transform as:
\beq T'_{a_1 a_2...a_n} = \lambda^{b_{a_1 a_2 ... a_n}} T_{a_1 a_2 ... a_n},~~ b_{a_1 a_2...a_n} = \sum_{i=1}^n(\delta_{a_i 0} - \delta_{a_i 1}),  \eeq
\noindent where $\delta_{ab}$ denotes the Kronecker delta symbol. The quantity, $b_{a_1 a_2 ... a_n}$, is called the {\it boost weight} (b.w.) of the frame component $T_{a_1 a_2 ... a_p}$. Any tensor can be decomposed in terms of the b.w. of its components and this b.w. decomposition gives rise to the {\it alignment
  classification} by identifying null directions relative to
which the components of a given tensor have a particular b.w. configuration. This classification reproduces the Petrov and Segre classifications in 4D, and also leads to a coarse classification in higher dimensions \cite{classa,classb,classc, OrtaggioPravdaPravdova:2013}. 

Denoting the boost order, $\mathcal{B}_{{\bf T}}(\bl)$, as the maximum b.w.~of a tensor, ${\bf T}$, for a null direction $\bl$. The Weyl tensor and any rank two tensor, {\bf T}, can be broadly classified into five {\it alignment types}: {\bf G}, {\bf I}, {\bf II}, {\bf III}, and {\bf N} if there exists an $\bl$ such that $\mathcal{B}_{{\bf T}} (\bl) = 2, 1, 0,-1,-2$, respectively, and we will say $\bl$ is ${\bf T}$-aligned, while if {\bf T} vanishes, then it belongs to alignment type {\bf O}. For higher rank tensors, like the covariant derivatives of the curvature tensor, the alignment types are still applicable despite the possibility that $|\mathcal{B}_{{\bf T}} (\bl)|$ may be greater than two. Any $\mathcal{I}$-degenerate spacetime admits a null frame such that all of the positive b.w. terms of the curvature tensor and its covariant derivatives are zero in this common frame, that is they are all of alignment type {\bf II} \cite{Hervik2011}. 

A degenerate Kundt spacetime admits a geodesic, shear-free, twist-free, non-expanding null congruence, $\bl$ for which the curvature tensor and its covariant derivatives are all, at least,  of type {\bf II} relative to $\bl$. In three and four dimensions, it has been proven that all $\mathcal{I}$-degenerate spacetimes are degenerate Kundt spacetimes, while in higher dimensions it is conjectured that the degenerate Kundt spacetimes contain all $\mathcal{I}$-degenerate spacetimes. Regardless of this conjecture, for any $\mathcal{I}$-degenerate spacetime where the curvature tensor and its covariant derivatives are of proper type {\bf II} (i.e., not of type {\bf D}), the metric cannot be characterized by $SPIs$, as only the b.w. zero components are involved in forming $SPIs$ \cite{CHP2010} and the negative b.w. components will never appear in the $SPIs$. Within the set of degenerate Kundt spacetimes, there is potentially a subclass of metrics that are in fact characterized by their $SPIs$, although not uniquely. For such a spacetime the curvature tensor and all of the covariant derivatives of the curvature tensor must be of alignment type {\bf D}. These spacetimes are necessarily $\mathcal{I}$-non-degenerate because any metric deformation within this subclass of metrics gives a new metric which has a differing set of $SPIs$. 

In this paper we investigate the class of spacetimes for which the Riemann tensor and its covariant derivatives are all of type {\bf D} relative to a preferred frame, which we will refer to as {\it type ${\bf D}^k$ spacetimes}. The following result was proven for type ${\bf D}^k$ spacetimes \cite{CHPP2009}:

\begin{thm}
If the Riemann tensor and all of its covariant derivatives $\nabla^{(k)}(Riem)$ of an $N$-dimensional spacetime $(M, {\bf g})$ are simultaneously of type {\bf D} (in the same frame), then the spacetime is degenerate Kundt. 
\end{thm}

\noindent In fact, it can be proven that any such metric admits two null geodesic expansion-free, shear-free, and twist-free vector fields, implying that the metric is doubly degenerate Kundt \cite{CHPP2009}. These spacetimes constitute an exceptional subcase of the degenerate Kundt spacetimes.

While this result restricts the candidates for type ${\bf D}^k$ spacetimes it does not provide an explicit form for the metrics. In 4D the boost isotropic spacetimes have been fully determined \cite{cahen1968}. In higher dimensions, five dimensional (5D) examples have been provided as Einstein solutions to supergravity  for which all SPIs are constant \cite{CFH2008}, and more generally all locally homogeneous type ${\bf D}^k$ spacetimes have been determined in arbitrary dimension \cite{SM2018}. To identify the Kundt metrics which are of type ${\bf D}^k$, we will integrate the equations given  by the vanishing of the positive and negative b.w. terms of the Riemann tensor and its covariant derivatives using the following theorem from \cite{CHP2010} and two lemmas: 

\begin{thm} \label{DKthrm}
If a spacetime is of type ${\bf D}^k$, then the components of the curvature tensor are determined by their $SPIs$.
\end{thm}

\begin{lem} \label{DKframes}
If the curvature tensor, $R_{abcd}$, and its covariant derivatives up to order $q \geq 1$, $R_{abcd;f_1,...,f_q},$  have only b.w. zero components, then the non-zero components of the Weyl tensor, $C_{abcd}$, and the Ricci tensor, $R_{ab}$, and their respective covariant derivatives up to order $q$ will be of b.w. zero.    
\end{lem}

\begin{lem} \label{Dklem}
For a type ${\bf D}^k$ spacetime, any $SPI$, $I \in \mathcal{I}$, must satisfy:
\beq l^a \partial_{x^a} I = n^a \partial_{x^a} I = 0. \nonumber \eeq 
\end{lem}

\section{Kundt geometries}
\label{sec_geom}

The line element of the most general non-twisting $N$-dimensional geometry can locally be written in the form \cite{Podolsky:2014}:
\begin{equation}
\d s^2 = g_{pq}(v,u,x)\, \d x^p\,\d x^q-2\,w_{p}(v,u,x)\, \d u\, \d x^p -2\,\d u\,\d v-2 H (v,u,x)\, \d u^2 , \label{general nontwist}
\end{equation}
where $p$ and $q$ range over the $N-2$ spatial indices. The non-twisting character of the spacetime implies the existence of a foliation by null hypersurfaces ${u=\,}$ constant (i.e., a family of maximal integral submanifolds labelled by the coordinate $u$). By the Frobenius theorem, this is equivalent to the existence (locally) of a non-twisting null vector field $\bl$ that is everywhere tangent (and normal) to ${u=\,}$constant. 

Since this vector field $\bl$ generates a congruence of null geodesics in the whole spacetime, it is most natural to take their affine parameter~$v$ as the second coordinate, so that $\bl=\mathbf{\partial}_v$. Locally, for any fixed $u$ and $v$ we are left with an ${(N-2)}$-dimensional Riemannian manifold covered by the spatial coordinates ${x^p}$. We will use the indices $m,n,p,q$ (ranging from $3$ to ${N-2}$) to label these spatial coordinates on the transverse space.

The non-vanishing contravariant metric components are \begin{equation}
g^{pq}\,, \quad g^{vu}=-1\,, \quad g^{vp}= g^{pq}w_{q} \,,
\quad g^{vv}= -2 H+g^{pq}w_{p}w_{q} \,, \label{ContravariantMetricComp}
\end{equation}
where $g^{pq}$ is an inverse matrix to $g_{pq}$. This implies that \cite{PS2015}:
\begin{equation}
w_{p}= g_{pq}g^{vq} \,, \quad 2H= -g^{vv}+g_{pq}g^{vp}g^{vq} \,. \label{CovariantMetricComp}
\end{equation}

The covariant derivative of the geometrically privileged null vector field ${\bl=\mathbf{\partial}_v}$ with respect to the metric~(\ref{general nontwist}) is $${\ell_{a;b}=\Gamma^{u}_{~ab}=\frac{1}{2}g_{ab,v}},$$ so that ${\ell_{v;b}=0=\ell_{a;v}}$. Denoting ${m_i^{~a}}$ as the components of the ${N-2}$ unit vector fields $\boldm_i$ which form an orthonormal basis in the transverse Riemannian space, we can show $m_i^{~u}=0$ using the condition $\boldm_i\cdot\bl=0$, and hence the \emph{optical matrix}~\cite{OrtaggioPravdaPravdova:2013} defined as ${\rho_{ij}\equiv \ell_{a;b}\,m_i^a m_j^b}$ takes the form: \beq {\rho_{ij}=\ell_{p;q}\,m^p_im^q_j= \frac{1}{2} g_{pq,v}\,m^p_im^q_j}. \nonumber \eeq This can be decomposed as ${\rho_{ij}=A_{ij}+\sigma_{ij}+\Theta\delta_{ij}}$, where ${A_{ij}\equiv\rho_{[ij]}}$ is the anti-symmetric \emph{twist matrix}, ${\sigma_{ij}}$ is the symmetric traceless \emph{shear matrix}, and the scalar ${\Theta\equiv\frac{1}{N-2}\,\delta^{ij}\rho_{ij}}$ determines the \emph{expansion} of the privileged vector field $\bl$. Requiring that the metric is also expansion-free and twist-free, we recover the \emph{Kundt class} of spacetimes, in which the transverse metric, \beq {g_{pq}(u,v,x^p)=g_{pq}(u,x^p)}, \nonumber \eeq is independent of the affine parameter $v$ \cite{PS2015}.

\section{Connection Coefficients and Curvature Components}
Expressing the metric relative to a null frame as:
\beq g_{ab} = - 2\ell_{(a} n_{b)} + m^i_{(a} m^j_{b)} \delta_{ij}, \label{etametric} \eeq
\noindent the simplest null frame for the general metric (\ref{general nontwist}) is given by
\begin{equation}
{\textstyle \bl=\mathbf{\partial}_v \, , \ \quad \boldn=-H \mathbf{\partial}_v+\mathbf{\partial}_u \, , \ \quad \boldm_i=m_i^{~p}(-w_{p}{\partial}_v+\mathbf{\partial}_p)} \,, \label{nat null frame}
\end{equation}
\noindent where we will write $W_i = m_i^{~p} w_{p}$ to express the corresponding coframe as:  
\beq \bl = -\d u\, , \ \quad \boldn = -\d v - H \d u - W_i m^i_{~p} \d x^p \, , \ \quad \boldm^i = m^i_{~p}\d x^p, \label{nat null coframe}
 \eeq
with the following non-vanishing inner products: ${\bl\cdot\boldn=-1}$, ${\boldm_i\cdot\boldm_j=\delta_{ij}}$.  For the degenerate Kundt spacetimes, the null vector field, $\bl$, will already be aligned with the Riemann tensor and its covariant derivatives, since they must be of type ${\bf II}$ to all orders.  It is possible that the null vector field, $\boldn$, is not yet aligned with the curvature tensors. We suppose there is a null rotation about $\bl$ such that the Riemann tensor is of alignment type ${\bf D}$:
\beq \tilde{\bl} &=& \bl, \nonumber \\
\tilde{\boldn} &=& {\boldn} + B_i {\boldm}^i + \frac{|B|^2}{2} {\bl}, \label{NullRot} \\  
\tilde{\boldm}^i &=& {\boldm}^i + B^i {\bl}. \nonumber \eeq
\noindent Here, the vector fields without tildes are the original frame and $B^i$ denotes the $i$-th null rotation parameter.

We will compute the connection coefficients and curvature components in the new frame:
\beq \begin{aligned} \tilde{\boldn} = -\boldsymbol{\theta}^1 &=  -\d v - \left(H+\frac{|B|^2}{2} \right) \d u -  ( W_i - B_i) m^i_{~p} \d x^p \\ &= - \d v - \hn \d u - \wn_i m^i_{~p} \d x^p, \\
\tilde{\bl} = -\boldsymbol{\theta}^2 &= - \d u, \\
\tilde{\boldm}^i = \boldsymbol{\theta}^i &= m^i_{~p} \d x^p - B^i du, \end{aligned} \label{NullRotV} \eeq

\noindent with corresponding dual vector fields:
\beq \begin{aligned} {\bf e}_1 &= \partial_v, \\
{\bf e}_2 &=  \partial_u + [-\hn - B^i \wn_i] \partial_v + B^im_i^{~p} \partial_{x^p}, \\
{\bf e}_i &= m_i^{~p} \partial_{x^p} - \wn_i \partial_v. \end{aligned} \eeq

% xxyyzz Comment on e_i vs m_i choice

From theorem 7.2 in \cite{CHPP2009}, the type ${\bf D}^k$ condition requires  that certain connection coefficients of b.w. $\pm 2$ and $\pm 1$ relative to this coframe must vanish. In particular, the connection coefficients: 
$$\Gamma_{22i}, \Gamma_{2ij},\text{ and } \Gamma_{i2j}$$ \noindent must vanish. We will assume that the frame fields are uniquely defined at all points of spacetime. When multiple pairs of null vector fields exists $(\bl,\bn)$ satisfying the type ${\bf D}^k$ property it is possible that other frames exists where the curvature tensor and its covariant derivatives will be of type {\bf D} without all connection coefficients vanishing; this issue is discussed further in section \ref{sec:Commutator}.

Cartan's first structure equation with respect to $\boldsymbol{\theta}^1 = - \boldn$ gives the following components:
\beq \begin{aligned} -2 \Gamma^1_{~[12]} &= {\bf e}_1(\hn) - {\bf e}_1(\wn_i) B^i, \\
-2\Gamma^1_{~[1i]} &= {\bf e}_1(\wn_i), \\
-2\Gamma^1_{~[2i]} &= {\bf e}_2(\wn_i) - B^j {\bf e}_j(\wn_i) - {\bf e}_i( \hn) - \tilde{w}_p m_{i,u}^{~p} + B^j \wn_{[i;j]}, \\
-2\Gamma^1_{~[ij]} &= \wn_{[i;j]}. \end{aligned} \label{dtheta1}\eeq

\noindent Similarly for $\boldsymbol{\theta}^i = \boldm^i$ we find:
\beq \begin{aligned} -2\Gamma^i_{~[12]} & = {\bf e}_1(B^i), \\
-2\Gamma^i_{~[2j]} &= m^i_{~p,u} m_j^p + 2 m^i_{~[p,q]}m_j^{~p} m_k^{~q}  B^k + {\bf e}_j( B^i), \\ 
-2\Gamma^i_{~[jk]} & = - 2m^i_{~[p,q]} m_j^{~p} m_k^{~q} = D^i_{~jk}.  \end{aligned} \label{dmi} \eeq

\noindent Since Cartan's first structure equation for $\boldsymbol{\theta}^2 =- \bl$  vanishes, the quantities:  $\Gamma^1_{~[2i]}$, $\Gamma^1_{~[ij]}$ and $\Gamma^i_{~[2j]}$ must be equal to zero, and so the remaining non-zero connection coefficients ($\Gamma_{abc} = -\Gamma_{cba}$) are: 
\beq  \Gamma_{221} &=& - [ {\bf e}_1(\hn) - {\bf e}_1(\wn_i) B^i ], \label{GammaS}\\
\Gamma_{21i} &=& \frac12 {\bf e}_1 (\wn_i + B_i), \\
\Gamma_{i21}  &=& -\frac12 {\bf e}_1(\wn_i - B_i), \\
\Gamma_{1i2} &=& \frac12 {\bf e}_1(\wn_i - B_i), \\
\Gamma_{ijk} &=& \frac12 [ D_{ijk} - D_{jki} + D_{kij}]. \label{GammaE} \eeq

From Cartan's second structure equation we find the following set of non-zero Riemnann tensor components: 
\beq \begin{aligned} R_{ijlm} &= {\bf e}_{l} (\Gamma_{imj}) - {\bf e}_{m} (\Gamma_{ilj}) + \Gamma_{il}^{~~n} \Gamma_{nmj} - \Gamma_{im}^{~~n} \Gamma_{nlj} - 2 \Gamma_{ikj}\Gamma^k_{~[lm]}, \\ 
R_{1212} &= {\bf e}_1(\Gamma_{122}) - \Gamma_{12i} \Gamma^i_{~12} - 2 \Gamma_{1i2} \Gamma^i_{~[12]}, \\
R_{12ij} &= {\bf e}_j (\Gamma_{1j2}) - {\bf e}_i (\Gamma_{1i2}) - 2 \Gamma_{1k2} \Gamma^k_{~[ij]}, \\
R_{ij12} &= \Gamma_{21i} \Gamma_{12j} - \Gamma_{12i}\Gamma_{21j} - 2\Gamma_{ikj} \Gamma^k_{~[12]}, \\
R_{1i2j} &= - {\bf e}_j(\Gamma_{12i}) + \Gamma_{1j2} \Gamma_{12i} + \Gamma_{12k}\Gamma^k_{~ji}, \\
R_{2i1j} &= - {\bf e}_j(\Gamma_{21i}) + 2 \Gamma_{21i} \Gamma_{2[1j]} - \Gamma_{1j2} \Gamma_{21i} + \Gamma_{21k} \Gamma^k_{~ji}. \end{aligned} \label{bw0Riem} \eeq 

\noindent Along with a set of Riemann tensor components that must vanish:
\beq \begin{aligned} & R_{ij2k} = {\bf e}_2 (\Gamma_{ikj}), ~ R_{ij1k} = {\bf e}_1(\Gamma_{ikj}), \\
& R_{121i} = {\bf e}_1 (\Gamma_{1i2} ), R_{1i12} = {\bf e}_1(\Gamma_{12i}), \\ 
& R_{122i} = {\bf e}_2(\Gamma_{1i2}) - {\bf e}_i(\Gamma_{122}), \\
&R_{2i12} = - {\bf e}_2 (\Gamma_{21i}). \end{aligned} \eeq

\noindent Taking the full contraction of the Riemann tensor gives the Ricci scalar: 
\beq R = -2 [ R_{1212} + R_{2~1k}^{~k} + R_{1~2k}^{~k}]+ {^S} R, \eeq
\noindent where ${^S}R$ denotes the Ricci scalar associated to the transverse space metric $g_{pq}$. 

The vanishing of $R_{121i} = R_{1i12} = 0$, which are invariant under a null rotation about $\bl$, implies that $W_i$ must be  linear in $v$, 
\beq W_i = W^1_i (u,x^p) v + W^0_i(u,x^p). \label{WBini} \eeq
\noindent Applying the Lie derivative with respect to $\bl$ to the Ricci scalar, we find that $H_{,vvv} = 0$ and so $H$ is of the form:
\beq H = H^{(2)} \frac{v^2}{2}  + H^{(1)} v + H^{(0)}. \label{Hini}\eeq

\begin{prop} \label{prop:NoB}
For a type ${\bf D}^k$ spacetime, if any non-zero null rotation parameter in \eqref{NullRot} needed to align $\bn$ is independent of $v$, i.e., $B_i,~i \in [3, N]$ such that $B^i_{,v} = 0$, then there exists a coordinate system where $B^i = 0$.
 
%zzyyxx That is, relative to the new coordinate system the aligned frame $\{\tilde{\bl}, \tilde{\boldn}, \tilde{\boldm}^i\}$ is a Kundt frame of the form \eqref{nat null coframe}.
\end{prop}

\begin{proof} 
At zeroth order, the non-zero components of the Riemmann tensor in equation \eqref{bw0Riem} will not contain the function $B^i(u,x^p)$. For higher order covariant derivatives of the Riemann tensor, we note that the connection coefficients in equations \eqref{GammaS}-\eqref{GammaE} arising in the covariant derivatives of a type {\bf D} tensor (i.e., $\Gamma_{21i}, \Gamma_{i21}, \Gamma_{1i2}$ and $\Gamma_{ijk}$) do not contain $B^i(u,x^p)$ or its derivatives.

To consider the effect of the frame derivatives involved in the covariant derivatives of the curvature tensor, we examine the components of $R_{abcd}$ and the connection coefficients (with the exception of $\Gamma_{122}$ which will not appear in any component $R_{abcd;e}$). As these are independent of $v$, the modified functions $\hn$ and $\wn_i$ will not appear in  any frame derivatives, and $B^i$ can only appear as coordinate derivatives of $x^p$ in ${\bf e}_2$. However, frame derivatives with respect to ${\bf e}_2$ will only occur in the components of the covariant derivatives of the curvature tensor with non-zero b.w., which must vanish due to the type ${\bf D}^k$ condition. 
 
Thus, the function $B^i(u,x^p)$ does not appear in any Cartan invariant at any order. This implies that the type ${\bf D}^k$ metric is equivalent to a type ${\bf D}^k$ metric where the null rotation from the Kundt frame to the aligned null vector field $\tilde{\boldn}$ does not require $B^i$. 

%zzyyzz
% where the aligned null vector $\tilde{\boldn}$ does  belongs to a Kundt frame already.

\end{proof}

Since the curvature components in the original frame are polynomials in $v$, if a non-trivial null rotation is required to align $\tilde{\boldn}$, the choice of null rotation parameters must depend on $v$. We will determine the form of the null rotation parameters, $B^i$ using the following two lemmas.

\begin{lem} \label{lem:KundtPoly}
For any degenerate Kundt spacetime, the covariant derivative of the curvature tensor of order $m$ can be decomposed into b.w. terms up to b.w. $-(m+2)$, and any component of b.w. $-p$ is a polynomial of degree $p$ in $v$.
\end{lem}

\begin{proof}
Since the metric function $H$ is quadratic in $v$ and the one-form components $W_i$ are linear in $v$, the proof follows by induction using the form of the connection coefficients relative to the Kundt frame.

Denoting the Riemann tensor as ${\bf R}$, for $m=0$ this is trivially true, while for $m=1$ the components of $\nabla {\bf R}$ can be decomposed as combinations of:

\begin{itemize}
\item b.w. 0: $(R)_0 (\Gamma)_0, D(R)_{-1}, m_i(R)_0$,
\item b.w. $-1$: $(\Gamma)_0 (R)_{-1}, (\Gamma)_{-1} (R)_0, D(R)_{-2}, m_i(R)_{-1}, {\bf e}_e(R)_0$,
\item b.w. $-2$: $(\Gamma)_0(R)_{-2}, (\Gamma)_{-1}(R)_{-1}, (\Gamma)_{-2} (R)_0, m_i (R)_{-2}, {\bf e}_e(R)_{-1}$,
\item b.w. $-3$: $(\Gamma)_{-1}(R)_{-2}, (\Gamma)_{-2} (R)_{-1} {\bf e}_e (R)_{-2}$,
\item b.w. $-4$: $(\Gamma)_{-2}(R)_{-2}$,
\end{itemize}
\noindent where $(T)_q$ denotes the components of a tensor ${\bf T}$ of b.w. $q$. The only components of the connection coefficients with b.w. $-2$ are $\Gamma^1_{~2i}$ which do not appear in combination with any component in the set $(R)_{-2}$, and so the b.w. $-4$ components do not appear. Writing the degree of a polynomial function of $v$ as $deg_v(~)$, we can confirm the inductive assumption by noting that $deg_v (\Gamma)_{0} = 0, deg_v (\Gamma)_{-1} = 1$ and $deg_v (\Gamma)_{-2} = 2$, with the action of the frame derivatives ${\bf e}_1$ and ${\bf e}_2$ on components, respectively, decrease and increase the degree of $v$ by one.

Assuming this holds for $m$, we will now prove for $m+1$. Writing ${\bf T} = \nabla_1 \hdots \nabla_{m+1} {\bf R}$ then the b.w. decomposition is

\begin{itemize}
\item b.w. 0: $(T)_0 (\Gamma)_0, D(T)_{-1}, m_i(T)_0$

\hspace{ 2 mm} {\vdots}

\item b.w. $-(m+1)$: $(\Gamma)_{-1} (T)_{-m}, (\Gamma)_{-2} (T)_{-m+1}, {\bf e}_2(T)_{-m}$.
\item b.w. $-(m+2)$:  $(\Gamma)_{-2} (T)_{-m}$.

\noindent Again, the formula for the covariant derivative of a type {\bf II} tensor ${\bf T}$  ensures there are no $(\Gamma)_{-2}$ components that combine with the $(T)_{-m}$ components, and so the most negative b.w. possible is $-(m+1)$ and $$deg_v (T)_{-(m+1)} = m+1.$$
\end{itemize}

\end{proof}

\noindent To continue we will introduce the concept of a {\it curvature matrix} and a {\it curvature vector} which are, respectively, a rank two tensor and a rank one tensor constructed from contractions of tensor products of the curvature tensor and its covariant derivatives. For both objects we will consider representations in terms of differential forms relative to the frame.

\begin{lem} \label{lem:Kundt:Operator}

For a degenerate Kundt spacetime, the components of b.w. $-p$ of any curvature matrix or curvature vector is a polynomial of degree $p$ in $v$.
\end{lem}

\begin{proof}
Given any tensor product from copies of $R_{abcd}$ and $R_{abcd;f_1,\hdots f_m}$ $m \geq 1$ with an even number of indices, only the components with b.w. 0, $-1$ and $-2$ remain after contracting indices to produce a curvature matrix. The components of the rank two tensor with fixed b.w. can be decomposed into products of the original curvature tensors of particular b.w. in a symbolic manner as:
\begin{itemize}
\item b.w. 0 : (b.w. 0),
\item b.w. -1 : (b.w. 0) (b.w. -1),
\item b.w. -2 : (b.w. 0) (b.w. -2) and (b.w. -1) (b.w. -1).
\end{itemize}

\noindent Then lemma \ref{lem:KundtPoly} implies that the  components of the original tensors with b.w. 0, $-1$ and $-2$ are polynomials of order 0, 1 and 2 in $v$. In the case of curvature vectors, the proof is similar except that we instead consider tensor products involving the curvature tensor and its derivatives with an odd number of indices, and lemma \ref{lem:CCoefs} is still applicable for the b.w. -1 components. 
\end{proof}

With these two lemmas we can determine the form of each null rotation parameter $B^i$.

\begin{prop} \label{prop:Bquad0}
If a spacetime is of type ${\bf D}^k$, then the non-zero null rotation parameters needed to align $\tilde{\boldn}$ are of the form: $$B^i = G(u,v,x^p) \phi^i(u,x^p) + B^{(1)~i}(u,x^p)v.$$
\end{prop}

\begin{proof}

%FACK Add discussion of curvature 'vectors' and how the combo of all curvature matrices and vectors being type D forces the space to be of type D^k
If a spacetime is of type ${\bf D}^k$, then there is a null rotation such that {\bf all} curvature matrices are of type {\bf D}. Using \eqref{NullRot}, the b.w. $-1$ and $-2$ components of an arbitrary curvature matrix, $R_{ab}$, are then:
\beq \begin{aligned}  \tilde{R}_{2i} &= R_{2i} + R_{12} B_i + 2R_{ij} B^j = 0, \\
 \tilde{R}_{22} &= R_{22} + \frac12 R_{12} |B|^2 + R_{2i} B^i + R_{ij}B^i B^j = 0. \end{aligned} \nonumber \eeq

\noindent From lemma \ref{lem:Kundt:Operator}, differentiating the components of $\tilde{R}_{2i}$ twice with respect to ${\bf e}_1 = \partial_v$ will cause $R_{2i}$ to vanish and give the following condition:
\beq R_{ij} B^j_{~,vv} = -\frac12 R_{12} B_{i,vv}. \nonumber \eeq

\noindent Thus, $B^i_{~,vv}$ is either an eigenvector of $R_{ij}$ or $B^1_{~,vv}$ vanishes. If $B^i_{,vv}$ is an eigenvector for {\bf all} curvature matrices then for each curvature matrix, the components are functionally dependent on $u$ and $x^p$. 

When all the curvature matrices are already of type {\bf D} form relative to the original frame, we can instead consider an arbitrary curvature vector, $T_a$, and determine a frame for which all curvature vectors are of type {\bf D} giving the condition: 
\beq 0 = B^i_{,vv} T_i. \nonumber \eeq
\noindent Either $B^i_{~,vv} = 0$ or $B^i_{~,vv}$ is orthogonal to all curvature vectors and it is dependent on $u$ and $x^i$ alone.

If the whole set of curvature matrices and vectors are of type {\bf D}, then we can treat the Weyl tensor as an operator and use the eigen-bivectors as invariant quantities. In such a case, either the frame is aligned with the Weyl tensor and it follows that the curvature tensor and its covariant derivatives are already of type {\bf D} form, or we can differentiate the null-rotated bivectors relative to the original frame to give $B^i_{,vv} \bm^i  = 0$ and the same argument can be repeated.
 
Thus, $B^i_{~,vv}$ must be of the form,
\beq B^i_{~,vv} = G(u,v,x^p) \phi^i(u,x^p), \nonumber \eeq

\noindent where $\phi^i$ is an eigenvector for all curvature matrices. Integrating twice we find 
\beq B^i = G(u,v,x^p) \phi^i(u,x^p) + B^{(1)~i}(u,x^p)v + B^{(0)~i}(u,x^p). \label{Biphi} \eeq

\noindent In analogy with the proof of proposition \ref{prop:NoB}, since the components $B^{(0)~i}$ never appear in any Cartan invariant, this implies that there is a coordinate transformation to a new Kundt frame where the required null rotation parameters have $B^{(0)~i} = 0$.

\end{proof}

Assuming $B^i \neq 0$ with $B^i_{,vv} \neq 0$, the type ${\bf D}^k$ condition imposes an additional constraint on the transverse metric, $g_{ij}$, and the metric functions $H$ and $W_i$:

\begin{prop} \label{prop:mi}
If a non-trivial null rotation with $B^i_{,v} \neq 0$ is required to align $\boldn$ with the type ${\bf D}^k$ frame, 
then the vector field $B^i m_i$ annihilates all b.w. 0 components of the curvature tensor and its covariant derivatives under frame differentiation.
\end{prop}

\begin{proof}

Any component of the $m$-th covariant derivative of the curvature tensor, $\nabla_1 \hdots \nabla_m {\bf R}$, relative to the aligned frame is of b.w. 0, and so the covariant derivative with respect to $\tilde{\boldn}$ must vanish in order to satisfy the type ${\bf D}^k$ condition. Furthermore, the vanishing of the connection coefficients $\Gamma_{22i}, \Gamma_{2ij}$ and $\Gamma_{i2j}$ implies that the covariant derivative with respect to $\tilde{\boldn}$ is in fact the frame derivative:

\beq {\bf e}_2 (\nabla_1 \hdots \nabla_m {\bf R}) =  [ \nabla_1 \hdots \nabla_m {\bf R}]_{,u} + B^i m_i^{~p} [\nabla_1 \hdots \nabla_m {\bf R}]_{,p} = 0. \nonumber \eeq

\noindent  Differentiating each component with respect to $D = \partial_v$ we find that the vector fields,
$$ {\bf X} = B^i_{,v} \bm_i \text{ and } {\bf Y} = B^i_{,vv} \bm_i,$$ 
\noindent will both annihilate $\nabla_1 \hdots \nabla_m {\bf R}$. Due to the form of $B^i$,
$$B^i = G(u,v,x^p) \phi^i + B^{(1)i}(u,x^p) v, $$
\noindent  we have the following expression in terms of {\bf X} and {\bf Y}:
\beq B^i \bm_i = \left[ {\bf X} - \frac{G_{,v}}{G_{,vv}} {\bf Y} \right] v + \frac{G}{G_{,vv}} {\bf Y}. \nonumber \eeq
\noindent Note that if $G_{,vv} = 0$ then ${\bf Y} = 0$ and $B^i m_i = v {\bf X}$. Thus by linearity $B^i \bm_i$ will annihilate all b.w. zero components, i.e., $$B^i \bm_i (\nabla_1 \hdots \nabla_m {\bf R}) = 0.$$ 
\end{proof}

\begin{cor} \label{DKcor}
If a metric is of type ${\bf D}^k$,  then there exists (locally) a coordinate transformation $(v',u',x^{'p}) = (v,u, f^p(u;x^q))$ such that 
\beq g_{rs} = g'_{pq} \frac{\partial f^p }{\partial x^r} \frac{\partial f^q }{\partial x^s},~~g_{pq,u} = 0. \nonumber \eeq
\end{cor}

\begin{proof}

For the type ${\bf D}^k$ spacetimes, the vector field $B^i  m_i$ annihilates all of the b.w. 0 components of the curvature tensor and its covariant derivatives, since any $SPI$ is constructed from the b.w. 0 components, lemma \ref{Dklem} implies that any $SPI$ must be independent of the null coordinates $u$ and $v$,
\beq I_{,u} = I_{,v} = 0. \nonumber \eeq
\noindent Theorem \ref{DKthrm} allows us to choose a frame where the components of the transverse curvature tensor:
$$ R_{ijkl} \text{ and } R_{ijkl; i_1\cdots i_n } ,$$
\noindent are all independent of $u$ and $v$, the corresponding $SPI$s of the Riemannian metric, $g_{pq}$, are independent of $u$ and $v$ as well. Since a Riemannian metric is characterized uniquely by its $SPIs$ and different values of $u$ in $g_{pq}(u,x^p)$ give the same set of $SPIs$, this implies that the metric is independent of $u$.

\end{proof}

\begin{prop} \label{prop:Bquad}
If a spacetime is of type ${\bf D}^k$ with a unique null frame $\{ \bl, \tilde{\boldn}, \bm_i \}$, then the null rotation parameters, $B^i$, required to align $\tilde{\boldn}$ must be linear in $v$. 
\end{prop}

\begin{proof}

%FACK Reread proof here

From proposition \eqref{prop:Bquad0}, $B^i_{,vv}$ has the form,
%We recall from the proof of proposition \ref{prop:Bquad0} that all curvature matrices in a  type ${\bf D}^k$ spacetime must satisfy the following condition with respect to the $v$-derivatives of the null rotation parameters:
%\beq  B_{i,vv} R_{12} = R_{ij} B^i_{~,vv} = 0. \nonumber \eeq
%\noindent This condition is satisfied when 
$$B^i_{~,vv} = G_{,vv}(u,v,x^p) \phi^i (u,x^p),$$ 
\noindent where $\phi^i$ is an eigenvector of all curvature matrices, a vector orthogonal to all curvature vectors and orthogonal to the frame basis elements, $\bm^i$. As the components of each $\bm^i$ are independent of $u$ and
% where $\phi^i$ is a common eigenvector of all curvature matrices of b.w. 0. Furthermore, since 
the b.w. 0 components are characterized by $SPI$s, which are independent of $u$ and $v$, 
%the above eigenvector equation implies that 
$\boldsymbol{\phi}$ must be independent of $u$ as well. By integrating this expression, coordinates can be chosen so that $B^i$ is of the form: 
\beq B^i = G(u,v,x^p) \phi^i(x^p) + B^{(1)~i}(u,x^p)v.\nonumber \eeq 

As all b.w. 0 components vanish under frame differentiation with respect to ${\bf e}_2$, and since $\phi^i$ may be determined in terms of the b.w. 0 components of some invariant quantity dependent only on the $x^i$ coordinates, we may make a rotation in the spatial plane to set $\bm^3 \propto \boldsymbol{ \phi}$ while preserving $\Gamma_{i2j} = 0$. Doing so, the null rotation parameters take the form: 
$$ B_3 = G(u,v,x^p)+ B_3^{(1)} v, \text{ and } B_{i'} = B_3^{(i')} v,~~i' > 3,$$ 
\noindent where proposition \ref{prop:NoB} has been used to choose a coordinate system such that all non-zero $B^{i}$ are at least linear in $v$.

%\noindent  The above eigenvector condition forces $R_{23} $ to vanish for all curvature matrices. 

To continue we will examine the vanishing of the structure constants:  $$ \Gamma_{i[2j]}, \Gamma_{2[2i]}, \text{ and } \Gamma_{2[ij]},$$ 
\noindent relative to a particular coordinate system; this choice of coordinates is motivated by proposition \ref{prop:mi} and the fact that if $B^i_{~,vv}\neq 0$ then  $\boldsymbol{ \phi} \propto G^{-1}_{~,vv} \bm_3$ annihilates the b.w. 0 components of the curvature tensor and its covariant derivatives. In particular, the components of the curvature tensor of the transverse space $\tilde{R}_{ijkl}$ and its covariant derivatives are annihilated by this vector field, and so all $SPI$s constructed for the transverse space will also be annihilated. Choosing rectifying coordinates for the transverse space so that $\bm_3 \propto \boldsymbol{ \phi} =  \partial_{x^3}$, since all $SPI$s are independent of $x^3$, the transverse metric, $g_{pq}$, must be independent of $x^3$ as well. Relative to this coordinate system, the vielbein matrix $m_i^{~p}$ is independent of $x^3$ while the orthogonality condition on the coframe implies that $m^{i'}_{~3} = 0$ for $i' >3$. 

Taking $\Gamma_{i[2j]}$ in \eqref{dmi} and differentiating twice with respect to $v$ gives
\beq 2 m^i_{[p,q]} m_j^{~p} m_3^{~q} B^3_{~,vv} + ({\bf e}_j(B^i))_{,vv} = 0. \nonumber \eeq

\noindent For $i'>3$, it follows that $m^{i'}_{~p,3}m_3^{~3} =0$ and $m^{i'}_{~3} = 0$ due to the independence of $x^3$ and the form of $m_i^{~p}$,  and so the above equation becomes: 
\beq  (\tilde{W}_3 B^{i'}_{,v})_{,vv} = 0. \nonumber \eeq

\noindent This can only occur if $B^{i'}_{,v} =0$ or $\tilde{W}_{3,vv} =  0$. The second condition is equivalent to $B_{3,vv} = 0$ and implies that all rotation parameters are linear in v, at most. If $\tilde{W}_{3,vv} \neq  0$ and  $B^{i'}_{,v} =0$, then it follows from proposition \ref{prop:NoB} that new coordinates can be chosen where $B^{i'} =0$. 

Assuming $\tilde{W}_{3,vv} \neq  0$ and  $B^{i'} = 0$, we will determine additional conditions at the algebraic level by considering the effect of the null rotation about $\bl$ on the curvature tensor: 
\beq \begin{aligned} R_{122i}' &= R_{122i} - R_{1(k|2|i)}B^k - R_{12ki} B^k + B_i R_{1212}, \\
R_{2ijk}' &= R_{2ijk} + 2 R_{2i1[j}B_{k]} -R_{lijk} B^l + R_{12jk} B_i, \\
R_{2i2j}' &= R_{2i2j} + 2R_{122(i} B_{j)} + R_{2(ij)k} B^k + R_{kilj} B^k B^l - 2 R_{1l2(i} B_{j)} B^l \\ & \hspace{5 mm} + |B|^2 R_{1(i|2|j)} - 2 R_{12l(i}  B_{j)}B^l + R_{1212} B_i B_j. \end{aligned} \eeq

\noindent Noting that the b.w. $-p$ terms are polynomials of degree $p$, we may differentiate appropriately with respect to $v$, to derive conditions on the components of the curvature tensor relative to the original Kundt frame. At the algebraic level this gives
\beq & R_{1323} = R_{1212},~~R_{1(3|2|i')} = - R_{123i'}, & \nonumber \\
& R_{3i'j'k'}=0, R_{12j'k'}=0, R_{231k'} =-  R_{123k'} , R_{2i'1k'} = R_{3i'3k'}, & \nonumber \\
& R_{1(i'|2|j')} =0,~R_{132j'} = - R_{1(3|2|j')}. & \nonumber \eeq

\noindent In addition we find that the curvature tensor is of type {\bf D} already, i.e.,  $$R_{122i} =0, R_{2i2j} = 0.$$ 
\noindent 

Due to the $v$-polynomial form of the components, this process can be repeated for higher order covariant derivatives of the curvature tensors. The result of this process is a set of conditions on the b.w. 0 components with the negative b.w. terms all vanishing in the original Kundt frame. This implies that such a solution is already of type ${\bf D}^k$ and any null rotation with $B^3 \neq 0$, $B^{i'} = 0$ yields another null direction for which the curvature tensor and all of its covariant derivatives are of type {\bf D}. This is a contradiction and so $B^i$ must be linear in $v$.
\end{proof}

For the moment, we will focus on the type ${\bf D}^k$ spacetimes where the null vector fields $\bl$ and $\tilde{\bn}$ are unique. Since the spatial vector field ${\bf X} = B^i_{,v} \bm_i$ annihilates all b.w. zero components and hence all $SPIs$ of the transverse space \cite{Hervik2011}, there is a Killing vector field $\boldsymbol{\zeta}$ in the span of {\bf X}, and so  
\beq {\bf X} = X(u,x^p) \boldsymbol{\zeta}. \nonumber \eeq
Choose rectifying coordinates so that $\boldsymbol{\zeta} = \partial_{x^3}$, the transverse metric, $g_{pq}$, is now independent of $x^3$. We may then rotate the frame so that 
\beq \bm_3 \propto \boldsymbol{\zeta}.  \eeq
\noindent Since the parameters of the spatial rotation vanish under frame differentiation of ${\bf e}_2$, the condition that the connection coefficients $\Gamma_{i2j}$ vanish will be preserved. Relative to this frame, the null rotation will only have one non-zero  component, $B^3$ after applying a coordinate transformation to set $B^{(0)~i}(u,x^p)$ equal to zero: $$B^3 = B^{(1)~3} v, B^i =0. $$
\noindent For the remainder of the paper we will work relative to this choice of coframe when a null rotation is needed to align $\bl$ and $\bn$.

\section{Commutator Conditions of Type ${\bf D}^k$ Spacetimes} \label{sec:Commutator}

There are two possible subclasses of type ${\bf D}^k$ spacetimes, depending on whether the pair of aligned null directions $\bl$ and $\bn$ are unique or if there are non-trivial null rotations preserving the type ${\bf D}^k$ form. We will denote each of these broader cases as I and II. 

If multiple aligned null directions exist, then for each pair of aligned null directions there is a boost isotropy, and so locally there exists at least two distinct boosts isometries for which the corresponding infinitesimal generators ${\bf X}^1$ and ${\bf X}^2$ have a non-zero commutator:
\beq [ {\bf X}^1, {\bf X}^2 ] = {\bf Z}, \nonumber \eeq

\noindent where the vector field ${\bf Z}$ will be another Killing vector field. 

Lemma \ref{Dklem} requires that the type ${\bf D}^k$ spacetimes must admit at least three Killing vector fields with 2-dimensional (2D) timelike orbits \cite{CHP2010}, the existence of additional boost isotropies implies that the isotropy group must be a semi-simple Lie group acting faithfully and non-properly on the spacetime and theorem 1.8 in \cite{DMZ} is applicable: 

\begin{thm} 
Let $G$ be a semi-simple group with finite center and no local $SL_2(\mathbb{R})$-factor, acting isometrically, faithfully, and nonproperly on a Lorentz manifold $\tilde{M}$, then:
\begin{enumerate}
\item $G$ has a local factor $G_1$ isomorphic to $O(1,n)$ or $O(2,n)$ with $n\geq 3$;
\item There exists a Lorentz manifold $S$  isometric, up to a finite cover, to $dS_n$ or $AdS_{n+1}$, depending whether $G_1$ is isomorphic to $O(1,n)$ or $O(2,n)$, and an open subset of $\tilde{M}$ in which each $G_1$-orbit is homothetic to $S$; 
\item Any such orbit as above has a $G_1$-invariant neighbourhood isometric to a warped product $L \times_{w} S$, for $L$ a Riemannian manifold. 
\end{enumerate}
\end{thm} 

In this case, the type ${\bf D}^k$ spacetime must be the warped product of a Riemannian metric, ${\bf h}$, with either de Sitter space ($dS_{\tilde{n}}$) or Anti-de Sitter space ($AdS_{\tilde{n}}$)  with a metric of the form:
\beq \begin{aligned} {\bf g'} = e^{\Phi(x^{r'})} g_{ab}(x^c) \d x^a \d x^b + h_{p'q'}(x^{r'}) \d x^{p'} \d x^{q'}, \end{aligned} \label{WrpdMetric} \eeq
\noindent with $a, b,c \in [0, \tilde{n}]$, $p', q',r' \in [\tilde{n}+1, n]$ and $g_{ab}$ denotes the de Sitter or Anti-de Sitter metric \cite{SM2018}:
\beq & \bg =-2\d u\left[\d v-\frac{v^2}2\left(1+\tan^2x\right)du+2v\tan x\d x\right]+\d x^2+\sin^2x~\d \Omega^2_{\tilde{n}-1}, & \label{dSn} \eeq
\small
\beq & \bg=-2\d u\left[\d v-\frac{v^2}2\left(-1+\tanh^2x\right)du-	2v\tanh x\d x\right]+\d x^2+\sinh^2x~\d \Omega^2_{\tilde{n}-1}. & \label{AdSn} \eeq
\normalsize

The metric in \eqref{WrpdMetric} is not quite in Kundt coordinates, but we can make a coordinate transformation $\tilde{v} = e^{\Phi} v$ to put this metric into the form \eqref{general nontwist}:
\small
\beq & {\bf g'} = -2 \d u \left[ \d \tilde{v} + \frac{\tilde{v}^2}{2} e^{-\Phi} H \d u + \tilde{v} w_{p} \d x^p - \tilde{v} \Phi_{,p'} \d x^{p'} \right]   + g_{pq} \d x^p \d x^q + h_{p'q'} \d x^{p'} \d x^{q'},  & \nonumber \eeq
\normalsize
\noindent where the isotropy  group of ${\bf g'}$ determines the functions $H$ and $w_p$, and the metric with indices $p,q \in [3, \tilde{n}]$: 

\begin{itemize}
\item  {\bf Case I.1}:

\beq &H = -(1+\tan^2 x),~~w_2 = 2\tan(x),~~w_{p} = 0,~~ p\in [3,\tilde{n}], & \\ & \bg = \d x^2+\sin^2x\d \Omega^2_{\tilde{n}-1}. & \label{dSnKundt} \eeq 

\item {\bf Case I.2}:

\beq &H = -(-1+\tanh^2 x),~~w_2 = -2\tanh(x),~~w_{p} = 0,~~ p\in [3,\tilde{n}], & \\ & \bg = \d x^2+\sinh^2x\d \Omega^2_{\tilde{n}-1}. & \label{AdSnKundt} \eeq 

\end{itemize}

\noindent As these solutions admit boost isotropy, they are necessarily type ${\bf D}^k$ spacetimes. We note that for these warped product metrics, relative to the Kundt coframe \eqref{nat null coframe} the connection coefficients: $\Gamma_{2[ij]}$, $\Gamma_{i[2j]}$ and $\Gamma_{2[2i]}$ all vanish.

In the case that $[{\bf X}^1, {\bf X}^2] = 0$, the Killing vector fields for any two boost isotropies are proportional, and the desired type {\bf D} null frame is uniquely determined at each point. Imposing the type ${\bf D}^k$ condition on the connection coefficients, $\Gamma_{22i}, \Gamma_{2ij},\text{ and } \Gamma_{i2j}$, we note that the vanishing of the corresponding anti-symmetric terms in \eqref{dtheta1} and \eqref{dmi} give simpler constraints on the metric functions.

\begin{lem} \label{lem:CCoefs}
If a metric is of type ${\bf D}^k$ then, relative to the frame \eqref{NullRotV} the metric functions $H$ and $W_i$ must satisfy: 
\beq & \wn_{[i;j]} = 0, & \\ &  -{\bf e}_i(\hn) +{\bf e}_2(\wn_i) - B^3 {\bf e}_3(\wn_i)= 0, & \label{Gamma22i} \\ &
 2 m^3_{~[p,q]}m_j^{~p} m_3^{~q}  B^3 + {\bf e}_j( B^3)=0. & \label{Gamma2ij} \eeq
\end{lem}

Expanding the first condition we find $\wn_{[i;j]} = w_{[p;q]}$ relative to the coordinate basis, and so $\tilde{w}_{p} = \tilde{m}^i_{~p} \wn^i$ must be the spatial gradient of some linear function of $v$: $$\tilde{w}_{p} =  \tilde{m}^i_{~p}   \wn^i = f_{,p}^1 v + f_{,p}^0,$$ with coefficients dependent on $u$ and $x^p$. Without loss of generality, we can apply a coordinate transformation: 
\beq & (v',u', x^{'i}) = (v+h(u,x^p), u, x^i) \to H' = H+h_{,u},~~W_i' = W_i  + h_{,p} m_i^{~p}, & \label{CT1} \eeq

\noindent to set $\tilde{w}_p^{(0)} = f_{,p}^0 = 0$, giving the form:
 \beq \tilde{w}_{p} = f_{,p}  v. \label{wpeqn} \eeq
 
\noindent This coordinate transformation will preserve the form of $B_i = B^{(1)}_i v$ since $B^{(0)}_i$ and $\tilde{W}^{(0)}_i$ do not appear in the set of Cartan invariants and there must be a coordinate system where they are all zero. 

Using the form of $\tilde{w}_p$, the vanishing of \eqref{Gamma22i} leads to three differential equations for the $\hn^{(i)}$ in $\hn$:
\beq & \hn^{(2)}_{,p} =  f_{,p} \hn^{(2)}, & \nonumber \\
&\hn^{(1)}_{,p} = f_{,p,u}, & \label{Heqn} \\ 
& \hn^{(0)}_{,p} = - f_{,p} \hn^{(0)}. & \nonumber \eeq

\noindent The $v$-coefficient $\hn^{(1)}$ will not appear in any of the Cartan invariants once the frame has been aligned implying that a coordinate transformation exists such that $\hn^{(1)}$ can be set to zero while preserving $B^{(0)}_i=0$ and $\tilde{W}^{(0)}_i=0$. Doing so, the second equation in \eqref{Heqn} forces $f$ to be separable, i.e.,
\beq f(u,x^p) = f_0(u) + f_1(x^p). \label{fsep} \eeq

\noindent Noting that $m^i_{~p,3} =0$ and $m_{3}^{~p'} =0$ for $p'>3$, the third condition in the above lemma can be expanded to give a differential equation,
\beq v m_j^{~p}  [ - m^3_{~3,p} m_3^{~3}  B^{(1)~3} + B^{(1)~3}_{,p} - (f_{,p} ) B^{(1)~3}] = 0, \eeq

\noindent which has the following solution:
\beq B^{(1)~3} = m^3_{~3} e^{f} b^3(u).  \label{Bsoln} \eeq

\section{Type {\bf D} Conditions on the Riemann tensor and its derivatives}

Many of the components of the curvature tensor with non-zero b.w. values will vanish automatically. There is only one set of components that will give additional conditions:
\beq \begin{aligned} R_{122i} &= {\bf e}_2(\Gamma_{1i2}) - {\bf e}_i(\Gamma_{122} ), \end{aligned} \eeq

\noindent Setting these components to zero and expanding in terms of $v$ yields two pairs of equations:
\beq \begin{aligned} &   B^{(1)}_{~~~3,u}=0, \\ & B^{(1)~3}  ( m_i^{~p} m_3^{~3} f_{,3,p} - \delta_i^{~3} b e^f f_{,3} ) =0. \end{aligned} \label{R122i} \eeq

%\beq \begin{aligned} & B^{(1)~3} e^f ( f_{,3,p} - e^{f} m^i_{~p} m_i^{~3} f_{,3}) =0  , \\
%&   B^1_{p,u}=0; \end{aligned} \label{R122i} \eeq
\begin{itemize}

\item {\bf Case II.1}: If $B^{(1)~3}$ is zero, these equations are automatically satisfied and so the metric functions are of the form:
\beq H^{(2)} = h_2(u) e^f,~~H^{(1)} = 0,~~H^{(0)} = h_0(u) e^{-f}, w_{p} = f_{1,p}, \label{case2mf}\eeq
\noindent where $f$ is separable with respect to $u$ and $x^p$.

\item {\bf Case II.2}: If $B^{(1)~3}$ does not vanish, then the first equation in \eqref{R122i} gives a condition on the $u$-dependence of $B^{(1)~3}$ which may be expanded to give a differential equation:
\beq B^{(1)}_{~~~3,u} = [b^{3}_{,u} + b^{3} f_{0,u}] m^3_{~3} e^f = 0. \eeq

\noindent Solving for $b^{3}$ we have the following expression for $B^{(1)~3}$: 
\beq B^{(1)~3}  = m^3_{~3} b e^{f_1}, \label{C1B} \eeq

\noindent where $b$ is a constant of integration. 

The functions $\hn$ and $\wn_p$ can be expressed in terms of $f = f_0(u) + f_1(x^{\hat{p}})$: 
\beq \hn^{(2)} = h_2(u) e^{f_1},~~\hn^{(1)} = 0,~~\hn^{(0)} = h_0(u) e^{-f_1},~~\wn_i = \tilde{m}_i^{~p} f_{1,p}. \label{case1mf} \eeq 
\noindent The remaining equations in \eqref{R122i} imposes two differential equations on $f_1$
\beq \begin{aligned}  f_{1,3,3} &= b(m^3_{~3})^2e^{f_1} f_{1,3}, \\  f_{1,3,p'} &= 0,~p'>3. \end{aligned} \label{Case1id} \eeq

\noindent The second condition requires that the function, $f_1$, is separable, i.e., $f_1=f_2(x^3)+f_3(x^{p'})$. The other differential equation for $f_1$ will be examined in the next section.

\end{itemize}

\subsection{Higher Derivatives} \label{sec:Hderiv}
Now that the Riemann tensor is in type {\bf D} form, to ensure that the covariant derivative of the Riemann tensor is of type {\bf D} as well we will consider the frame derivatives with respect to ${\bf e}_1$ and ${\bf e}_2$. Due to the $v$-independence of the Riemann tensor components and the form of the non-zero connection coefficients we only need to consider the frame derivative of ${\bf e}_2$. To see why this is the case consider a tensor with $A,B,C$ denoting groupings of indices whose cumulative b.w. is zero, then the covariant derivative will be of the form:   
\beq T_{A 1 B 2 C;2} = T_{A 1 B 2 C,2} - \Gamma^1_{~21} T_{A 1 B 2 C} - \Gamma^2_{~22} T_{A 1 B 2 C} + \sum_{A,B,C} \Gamma \cdot T,\nonumber \eeq
\noindent where the last term is a symbolic expression denoting the remaining terms involved in the covariant derivative. Since  $\Gamma^1_{~21} = -\Gamma^2_{~22}$, the second and third terms will cancel. Repeating this process for every occurrence of $1$ and $2$ indices in $A,B$ and $C$ leads to the result 
\beq T_{A1B2C;2} = T_{A1B2C,2}. \nonumber \eeq

If the frame derivatives of the components of the curvature tensor with respect to ${\bf e}_2$ vanish, then the higher order covariant derivatives of the curvature tensor will be of type {\bf D}.  We will consider the non-zero b.w. 0 components:
\beq \begin{aligned} R_{1212} &= H_{,v,v} +\frac14 W_{i,v} W^{i,v},  \\
R_{1i2j} &= -\frac12 W_{i;j,v} +\frac14  W_{i,v} W_{j,v}, \\
R_{12ij} &= W_{[i;j],v}. \end{aligned}\eeq

\begin{itemize}
\item {\bf Case II.1:} In the case where $\bl$ and $\bn$ are already aligned, we only need to consider ${\bf e}_2 (R_{1212})$ to determine the form of $H$: 
\beq {\bf e}_2 (R_{1212}) = {\bf e}_2 (H_{,v,v}) = (h_2(u) e^{f}) = 0, \nonumber \eeq

\noindent solving the differential equations for $h_2$, $H^{(2)}$ will be of the form
\beq H^{(2)} = C e^{f_1}, \eeq

\noindent where $C$ is a constant of integration.

\item {\bf Case II.2: } The vanishing of ${\bf e}_2(R_{123i})$ leads to two branches of the following equation: 
\beq {\bf e}_2(R_{123i'}) = \frac12 {\bf e}_2 \left(  b e^{f_1}  m_{i'}^{~p'}  ( m^3_{~3} f_{1,p'} + 2 m^3_{~3,p'})  \right) = 0, \nonumber \eeq
\noindent or more simply 
\beq \frac{1}{2} v b^2 e^{2f_1} \left[ f_{1,x} m_{i'}^{~p'} (f_{1,p'} m^3_{~3} +2 m^3_{~3,p'}) \right] = 0. \label{CaseIIsplit} \eeq

%\beq {\bf e}_2(R_{123i'}) = {\bf e}_2 ( m^3_{~3} b e^{-f_1} m_{i'} (f_1) ) = b e^{-f_1} m_{i'} (f_1) m_3(f_1) + 2 b e^{-f_1} m_{i'}(f_1) m_3 (m^3_{~3}). \nonumber \eeq

\noindent Assuming that the separable function $f_1=f_2(x^3)+f_3(x^{p'})$ is non-constant: either $f_2=0$ or $f_3$ is of the form:
\beq f_3 = - 2 \ln ( m^3_{~3}) + C' \label{f_3} \eeq
\noindent where $C'$ is a constant of integration.

\begin{itemize}
\item {\bf Case II.2a}: If $f_{1,3} = 0$ then ${\bf e}_2(R_{12ij})= {\bf e}_2(R_{1i2j})=0$ and ${\bf e}_2(R_{1212}) = 0$ gives the same condition on $\hn^{(2)}$ as in Case II.1,
\beq \hn^{(2)} = C e^{f_1}.  \eeq

\item {\bf Case II.2b}: If $f_{1,p'} = 0$  then the differential equation \eqref{Case1id} gives the solution: 
\beq f_1(x^3) = -C' + C_0C_1+C_0 x^3+\ln\left( \frac{C_0}{1- e^{C_0C_1+C_0 x^3}} \right).  \eeq

%\beq \hn^{(2)}_{,u} = [h_{2,u} + h_2 f_{,u}]e^f = 0,~~f=f_0(u)+f_1(x^{\hat p}). \eeq

\noindent Substituting this into ${\bf e}_2 ( R_{1212}) = 0$ gives two equations as coefficients of $v$. The first will remove $u$ dependence from $\hn$,
\beq \hn^{(2)} = Ce^{f_1}, \eeq

\noindent where $C$ is a constant. while the second condition fixes the parameter, $b$: 
\beq b = -\frac{C}{C_0}. \eeq
\noindent Without loss of generality, the constant $C'$ can be set to zero as it does not appear in any metric function.

In this case, the metric functions have the following form:
\beq H = \frac{v^2}{2} \left[ \frac{ C C_0^2 e^{f} - C^2 e^{2f}}{C_0^2 (m^3_{~3})^2} \right],~~w_p = [f_{,p} + 2 \ln (m^3_{~3})_{,p} + \delta_{3p} (m^3_{~3})^2 b e^f ]v,\nonumber \eeq

\noindent with the function, $f$, \beq f(x^3) = C_0(C_1+ x^3)+\ln \left( \frac{C_0}{1- e^{C_0(C_1+ x^3)}} \right),\nonumber \eeq 
\noindent and the function $m^3_{~3}$ along with the components of the  transverse metric $g_{pq}$ are arbitrary functions of $x^{p'}$. Writing $\Phi  = 2 \ln( m^3_{~3})$, we can use the coordinate transformation $\tilde{v} = - e^{\Phi} v$ to rewrite the metric as a warped product of a Riemannian metric, $g_{p'q'}$ with a three-dimensional (3D) spacetime with metric: 
\small
\beq \begin{aligned} d\tilde{s}^2 - \d u \left( \d v + \frac{CC_0^3 e^{C_0(C_1 + x)}v^2}{(Ce^{C_0(C_1+x)}+C_0)^2} \d u + \frac{vC_0 (Ce^{C_0(C_1+x)}-C_0)}{Ce^{C_0(C_1+x)}+C_0} \d x \right) + \d x^2. \end{aligned} \nonumber \eeq
\normalsize
\noindent A simple calculation shows that the 3D spacetime is in fact $dS_3$. Thus, this type ${\bf D}^k$ metric belongs to Case I.1.

\end{itemize}   

\end{itemize}

As $\Gamma_{122}$ will not appear in any covariant derivative and the frame derivatives of ${\bf e}_1$ and ${\bf e}_2$ annihilate all components of the curvature tensor and the connection coefficients, we conclude that any higher covariant derivative of the Riemann tensor will be of type {\bf D}. In analogy with theorem 7.1 in \cite{CHPP2009}, where a Kundt spacetime will be degenerate Kundt if the curvature tensor and its first covariant derivative are of alignment type {\bf II}, we have shown the following result: 

\begin{cor} \label{cor:D1toDk}
If the Riemann tensor and its first covariant derivative are of type {\bf D}, then the spacetime is of type ${\bf D}^k$.
\end{cor}

\section{Conclusions}

We have determined all spacetimes admitting a (local) boost isotropy. This condition implies that the curvature tensor and its covariant derivatives must be invariant under a boost in at least one plane spanned by two null directions $\bl$ and $\bn$. Adapting a frame basis to this pair of null vector fields, the isotropy condition forces the curvature tensor and its covariant derivatives to be of type {\bf D} relative to this basis. Using the result that all type ${\bf D}^k$ spacetimes must belong to the degenerate Kundt class, we have shown that these metrics take a specific form depending on their isotropy group, and this can be summarized in the following proposition:

\begin{prop}

If the isotropy group of a spacetime contains a subgroup isomorphic to $O(1,\tilde{n})$, or $O(2,\tilde{n})$ with $\tilde{n}\geq 3$ acting isometrically, faithfully and nonproperly on the spacetime, then the metric is isometric (up to a finite cover) to a warped product of a Riemannian manifold ${\bf h}$ with $dS_{\tilde{n}}$ or $AdS_{\tilde{n}+1}$, respectively. 

\begin{itemize}
\item {\bf Case I.1}: For $dS_{\tilde{n}}$, $\tilde{n} \geq 3$, the  metric functions for the metric \eqref{dSn} are:
\beq \begin{aligned} & H^{(2)} = -e^{-\Phi} (1 + \tan^2 x),~ w_3^{(1)} = 2 \tan x, w_{p'}{(1)}  = \Phi_{,p'}, \\ & {\bf \tilde{g}} = e^{\Phi}( dx^2 + \sin^2 x d \Omega^2_{\tilde{n}-1}) + h_{p' q'} (x^{r'}) dx^{p'} dx^{q'}. \end{aligned} \label{Case0a} \eeq

\item {\bf Case I.2}: For $AdS_{\tilde{n}+1}$, $\tilde{n} \geq 3$, the metric \eqref{AdSn} has metric functions:
\beq \begin{aligned}  & H^{(2)} = -e^{-\Phi} (-1 + \tanh^2 x),~ w_3^{(1)} = 2 \tanh x, w_{p'}{(1)} = \Phi_{,p'}, \\
& {\bf \tilde{g}} =  e^{\Phi}( dx^2 + \sinh^2 x d \Omega^2_{\tilde{n}}) + h_{p' q'} (x^{r'}) dx^{p'} dx^{q'}. \end{aligned} \label{Case0b} \eeq
\end{itemize}

\noindent If the isotropy group contains a subgroup isomorphic to $O(1,1)$ acting nonproperly on the spacetime then the type ${\bf D}^k$ spacetimes consist of two different metric classes:

\begin{itemize} 
\item {\bf Case II.1}: The metric is isomorphic to a warped product of a Riemannian manifold ${\bf h}$ with $dS_1$ or $AdS_{2}$ respectively, i.e., the metric functions are of the form: \beq H = \frac{v^2}{2} Ce^{f} ,~~w_p = f_{,p}v,\nonumber \eeq

\noindent where the function $f=f(x^{p})$ and the transverse metric is arbitrary.

\item {\bf Case II.2a}: The Kundt metric has the following metric functions: \beq H = \frac{v^2}{2} [C e^{f} - g_{33} b^2 e^{2f}],~~w_p = [f_{,p} + \delta_{3p} (m^3_{~3})^2 b e^f ]v,\nonumber \eeq

\noindent where the function $f=f(x^{\hat{p}})$, $\hat{p} > 3$, and the components of the transverse metric $g_{pq}(x^{\hat{r}})$ are independent of $x^3$ but otherwise arbitrary.

\end{itemize}
\end{prop}

\noindent Case II.2b, introduced in section \ref{sec:Hderiv} is in fact a warped product of a Riemannian manifold with $dS_3$ and hence belongs to case I.1. 

We note that the type ${\bf D}^k$ spacetimes are contained in the class of degenerate Kundt spacetimes, and hence are $\mathcal{I}$-degenerate in this larger class. However, when restricted to the subclass of type ${\bf D}^k$ spacetimes, these spacetimes are $\mathcal{I}$-non-degenerate since these spacetimes are in fact locally characterized by their $SPIs$ \cite{CHPP2009}. Any restricted deformation of a given type ${\bf D}^k$ metric which preserves the ${\bf D}^k$ condition produces a new type ${\bf D}^k$ metric with a different set of $SPIs$.

In future work we will determine the existence of type ${\bf D}^k$ Einstein metrics in arbitrary dimension. Due to the difficulty in determining type {\bf D} Einstein spacetimes in higher dimensions \cite{pravdab, wylleman2015finalizing}, the type ${\bf D}^k$ metrics present a restricted subclass of spacetimes which may provide interesting Einstein metrics. However, there are some examples  of locally homogeneous Einstein type ${\bf D}^k$ solutions relative to the Kundt coordinates \cite{CFH2008}, which are reviewed in Appendix A. The conditions for a general type ${\bf D}^k$ metric to be Einstein are not readily solved in the current coordinate system. By choosing a new coordinate system adapted to the two geodesic shear-free, twist-free and expansion-free  null directions, we hope to determine the set of Einstein type ${\bf D}^k$ spacetimes. In Appendix B, we have summarized all 4D type ${\bf D}^k$ spacetimes, and since these metrics have a potential physical interpretation, we display the Ricci tensor for each subcase as well. This is meant to illustrate the difficulties related to this coordinate system, for example, when determining the existence of Einstein metrics.

\section*{Acknowledgements}

This work was supported through the Research Council of Norway, Toppforsk grant no. 250367: Pseudo-
Riemannian Geometry and Polynomial Curvature Invariants: Classification, Characterisation and Applications (S.H. L.W. and D.M.), and NSERC (A.C.).

\section*{Appendix A: Five Dimensional Einstein $CSI$ Type ${\bf D}^k$ Metrics}

%fack A says "wording" in the first five lines of this para.

Within the class of degenerate Kundt spacetimes, there a subclass of spacetimes for which all $SPIs$ are constant or zero, known as $CSI$ or $VSI$ spacetimes, respectively. The $CSI$ and $VSI$ spacetimes are of particular relevance in the context of  supergravity since they contain solutions to supergravity which also admit supersymmetries \cite{CFH2008}. A general approach to generate locally homogeneous Einstein Kundt-$CSI$ solutions was introduced in \cite{CFH2008}. This approach relies on a particular choice of the locally homogeneous transverse metric, ${\bf g}^\perp$, and was illustrated in 5D. Each of the 5D $CSI$ Einstein solutions admit a $CSI$ type ${\bf D}^k$ solution as a  special case. We will review the three 5D solutions that have been constructed and show that they belong to Case II.2a. 

\begin{itemize}
\item {\it Heisenberg group:} 

The Heisenberg group can be equipped with the left-invariant metric:
\beq g^\perp_{~pq} \d x^p \d x^q = \left( \d x + \frac{b}{2}( y \d z - z \d y)\right)^2 + \d y^2 + \d z^2, \eeq
\noindent and the non-zero metric function $H$ and one-form ${\bf w}  = w_p \d x^p$ are then
\beq H = \frac{b^2}{4}v^2+ H^{(0)}(u,x,y,z),~~{\bf w} = \sqrt{2} b \left(\d x + \frac{b}{2} ( y \d z - z \d y)\right). \label{Heisenberg}\eeq

\noindent This Kundt metric will be Einstein if $H^{(0)}$ satisfies, \beq \Box^\perp H^{(0)} + \left( H^{(0)}W_i \right)^{;i} = 0, \label{EinsteinC} \eeq
\noindent while the metric will be of type ${\bf D}^k$ if $H^{(0)} = 0$, and so the type ${\bf D}^k$ solution is automatically Einstein. To show it is of type ${\bf D}^k$, we complete the Kundt coframe \eqref{nat null coframe} with: 
\beq \bm^3 = \d x + \frac{b}{2}( y \d z - z \d y),~ \bm^4 = \d y, ~ \bm^5 = \d z, \eeq 
\noindent and apply a null rotation \eqref{NullRotV} with $B^3 = \sqrt{2} b$ and $B^{i'} = 0, i'>3$. The Ricci tensor is of the form:
\beq {\bf R} = -\frac{b^2}{2} ( - 2 \bl \bn + \delta_{ij} \bm^i \bm^j ). \eeq

\item {\it SL(2, $\mathbb{R}$):} 

Forming a Kundt metric from the left-invariant metric
\beq  g^\perp_{~pq} \d x^p \d x^q = \left( \d x - a \frac{\d z}{y}\right)^2 + \frac{b^2}{y^2}(\d y^2 + \d z^2), \eeq
\noindent with the non-zero metric function $H$ and one-form ${\bf w}  = w_p \d x^p$ are: 
\beq H = \frac{a^2}{4b^4}v^2 + H^{(0)}(u,x,y,z),~~ {\bf w} = \frac{\sqrt{2(a^2+b^2)}}{b^2} \left( \d x - a \frac{\d z}{y}\right) \label{sl2r}. \eeq
\noindent If $H^{(0)}$ satisfies \eqref{EinsteinC} for the current one-form, {\bf w} in \eqref{sl2r}, then the resulting Kundt metric is Einstein, while if $H^{(0)}$ vanishes then the metric is of type ${\bf D}^k$ as well. Using the Kundt coframe \eqref{nat null coframe} with: 
\beq \bm^3 = \left( \d x - a \frac{\d z}{y}\right), \bm^4 =\frac{b}{y} \d y,~~ \bm^5 = \frac{b}{y} \d z \eeq
\noindent a null rotation \eqref{NullRotV} with $B^3 = \frac{\sqrt{2(a^2+b^2)}}{b^2}$ and $B^{i'} = 0, i'>3$ puts the curvature tensor and its covariant derivatives into type {\bf D} form and the Ricci tensor is of the form:
\beq {\bf R} = \frac{-(a^2+2b^2)}{2b^4} ( - 2 \bl \bn + \delta_{ij} \bm^i \bm^j ). \eeq

\item {\it The 3-sphere $S^3$:} 

Taking the Berger metric as the transverse space metric,
\beq g^\perp_{~pq} \d x^p \d x^q = a^2 (\d x + \sin y \d z)^2 + b^2 (\d y^2 + \cos^2 y \d z^2), \eeq
 \noindent with the non-zero metric function $H$ and one-form ${\bf w}  = w_p \d x^p$: 
\beq H = \frac{a^2}{4b^4}v^2 + H^{(0)}(u,x,y,z),~ {\bf w} = \frac{a \sqrt{2(a^2-b^2)}}{b^2} (\d x + \sin y \d z). \label{3sphere} \eeq
\noindent As in the previous two examples, if $H^{(0)}$ satisfies \eqref{EinsteinC} for the current one-form, {\bf w} in \eqref{3sphere}, then the metric is Einstein, and if $H^{(0)}$ vanishes then the metric is of type ${\bf D}^k$. Using the Kundt coframe \eqref{nat null coframe} with: 
\beq \bm^3 = (\d x + \sin y \d z), \bm^4 = b \d y,~~ \bm^5 = b \cos y \d z, \eeq
\noindent a null rotation \eqref{NullRotV} with $B^3 =\frac{a \sqrt{2(a^2-b^2)}}{b^2}$ and $B^{i'} = 0, i'>3$ shows that the curvature tensor and all of its covariant derivatives are of type {\bf D} form, and the Ricci tensor is:
\beq {\bf R} = \frac{-(a^2-2b^2)}{2b^4} ( - 2 \bl \bn + \delta_{ij} \bm^i \bm^j ). \eeq

\end{itemize}

\section*{Appendix B: The Ricci Tensor of the 4D ${\bf D}^k$ Metrics }

In 4D, the boost isotropic metrics have been entirely determined, albeit in a different coordinate system  \cite{cahen1968}. Using the isometry groups as a coarse classification we can determine relationships between the 4D metrics we have found in this paper with those in \cite{cahen1968}. The metrics in Case II.1 correspond to the metrics in equation 2.9 of \cite{cahen1968}:
\beq \d s^2 = (1+\lambda x)^2 \frac{\d u \d v}{\left(1+\frac{K}{4}uv\right)^2} -\alpha^2(x,t) \d t^2 - \beta^2(x,t) \d x^2. \label{Cahen1} \eeq 
\noindent Case II.2a corresponds to the metrics in equation 2.11 in \cite{cahen1968}: 
\beq \d s^2 = q^2(x) \frac{\d u \d v}{\left(1+\frac{K}{4}uv\right)^2} - \frac{f(x)^2}{4}\left[ \d t - \frac{C_1 v\d u}{1-\frac{K}{4}uv}+ \frac{C_1 u \d v}{1 - \frac{K}{4}uv }\right]^2 - \frac{\d x^2}{f^2(x)} \label{Cahen2} \eeq
\noindent In this coordinate system, $(u,v,t,x)$, $u$ and $v$ are null coordinates while $t$ and $x$ are spatial coordinates. The type ${\bf D}^k$ metrics in Case I.1 and I.2 arising from warped products of $dS_{n}$, $2 \leq n \leq 4$ or $AdS_{n'+1}$, $1 \leq n' \leq 3$ with a 2D Riemannian metric will belong to subcases of the metrics in \eqref{Cahen1} and \eqref{Cahen2}.

To continue we will, consider the 4D type ${\bf D}^k$ metrics, and for each case we will compute their Ricci tensor components using the formulas:
\beq \begin{aligned} R_{12} &= H^{(2)} - \frac12 W^{(1)~i}_{~~~;i} + \frac12 W^{(1)~i} W_i^{(1)}, \\
 R_{ij} &= - W^{(1)}_{(i;j)} + \frac12 W^{(1)}_i W^{(1)}_j + \tilde{R}_{ij}, \end{aligned} \eeq
\noindent where $\tilde{R}_{ij}$ denotes the Ricci tensor of the transverse space. For the ${\bf D}^k$ metrics, the Ricci tensor can have Segre type: $[11,(1,1)], [(11),(1,1)], [1,(11,1)]$ and $[(111,1)]$. The Ricci tensor has a physical interpretation as a non-null Maxwell field for Segre type $[(11),(1,1)]$ and either a vacuum or a cosmological constant for Segre type $[(111,1)]$ In order for a metric to be Einstein, the Ricci tensor must be of the form $R_{ab} = \lambda g_{ab}$ where $\lambda$ is a constant, and so
\beq R_{12} = -\lambda,~~ R_{ij} = \delta_{ij} \lambda. \label{EinsteinM} \eeq

\noindent We note that in GR, the only solutions admitting a Lie group of six Killing vector fields, $G_6$, with 3D timelike orbits, $T_3$, will be vacuum or $\Lambda$ solutions. Thus these are the only spaces of constant curvature \cite{exact}. Nonetheless we will provide the Ricci tensors for the warped products of $dS_3$ or $AdS_3$ with $\mathbb{R}$. 

\begin{itemize}

\item {\bf Case I.1}: 

Using the coordinates $(u,v,x,y)$, the line-element is: 
\small
\beq \begin{aligned} \d s^2 &=-2\d u\left[\d v-\frac{v^2}2 e^{-\Phi} \left(1+\tan^2x\right)\d u+ 2 v\tan x \d x - v \Phi_{,y} \d y] \right]+\d x^2+ m^4_{~4}(y) \d y^2, \end{aligned} \eeq
\normalsize

\noindent where $\Phi$ and $m_4^{~4}$ are arbitrary functions of $y$. Then relative to the type ${\bf D^k}$ null frame, $\{ \bn, \bl, \bm^3, \bm^4\}$ with $i, j \in [3,4]$ the Ricci tensor components are:
\beq \begin{aligned}  R_{12} &= \frac12 m_4^{~4}( m_4^{~4} \Phi_{,y})_{,y} + \frac12 (m_4^{~4} \Phi_{,y})^2   - 2 - (1+\tan^2 x) e^{-\Phi} + (1+\tan^2 x), \\
R_{33} &= 2, \\
R_{34} &= m_4^{~4} \Phi_{,y} \tan x, \\
R_{44} &= - m_4^{~4} ( m_4^{~4} \Phi_{,y})_{,y} - \frac12 (m_4^{~4})^2 \Phi_{,y}^2 .    \end{aligned} \eeq

\noindent To be an Einstein manifold, $R_{34}$ must be zero, so that the warping function is constant. This forces $R_{44}$ to vanish, so that the Ricci tensor cannot satisfy the Einstein condition \eqref{EinsteinC}. 

\item {\bf Case I.2}: 

Relative to the coordinates $(u,v,x,y)$, the line element is:
\small
\beq \begin{aligned} \d s^2 & =-2\d u\left[\d v-\frac{v^2}2 e^{-\Phi}\left(-1+\tanh^2x\right)\d u-	2v\tanh x\d x - v \Phi_{,y} dy] \right]+\d x^2+ m^4_{~4}(y) \d y^2, \end{aligned} \eeq
\normalsize

\noindent where $\Phi$ and $m_4^{~4}$ are arbitrary functions of $y$. Then with respect to the type ${\bf D^k}$ null frame, $\{ \bn, \bl, \bm^3, \bm^4\}$ with $i, j \in [3,4]$ the Ricci tensor components are:
\beq \begin{aligned}  R_{12} &= \frac12 m_4^{~4}( m_4^{~4} \Phi_{,y})_{,y} + \frac12 (m_4^{~4} \Phi_{,y})^2    2 - (-1+\tanh^2 x) e^{-\Phi} + (-1+\tanh^2 x), \\
R_{33} &= -2, \\
R_{34} &= -m_4^{~4} \Phi_{,y} \tanh x, \\
R_{44} &= - m_4^{~4} ( m_4^{~4} \Phi_{,y})_{,y} - \frac12 (m_4^{~4})^2 \Phi_{,y}^2 .    \end{aligned} \eeq

\noindent To be an Einstein manifold, $R_{34}$ must be zero, and so the warping function must constant. With $\Phi$ constant, $R_{44}$ vanishes and this spacetime cannot be Einstein.
 
\item {\bf Case II.1 }: 

In the coordinate system, $(u,v,x^3,x^4)$, the line element is:
\beq \d s^2 = - \d u \left( \d v - C e^{\Phi} \frac{v^2}{2} \d u - \Phi_{,x^3} \d x^3 - \Phi_{,x^4} dx^4 \right) + g_{pq}(x^3,x^4) \d x^{p} \d y^{q},  \eeq
\noindent where $C = -1$ or $1$ for the 2D de Sitter or anti-de Sitter space respectively, the indices $p, q \in [3,4]$ and $\Phi$ is an arbitrary function of $x^p$.  Then relative to the type ${\bf D^k}$ null frame, $\{ \bn, \bl, \bm^3, \bm^4\}$ with $i, j \in [3,4]$ the Ricci tensor components are:
\beq \begin{aligned}  R_{12} &= - C e^{-\Phi} + \frac12 [ m_{i}^{~p} \Phi_{,p} m^{i~q} \Phi_{,q} - m_{i}^{~q} (m_{i}^{~p} \Phi_{,p})_{,q} + \Gamma^{k}_{~ij} \delta^{ij} m_{k}^{p} \Phi_{,p}], \\
 R_{ij} &= m_{(i}^{~q} (m_{j)}^{~p} \Phi_{,p})_{,q} - \Gamma^{k}_{~(ij)} m_{k}^{p} \Phi_{,p} + \frac12  m_{i}^{~p} m_{i}^{~q} \Phi_{,q} \Phi_{,p} + \tilde{R} \delta_{ij} .    \end{aligned} \eeq

\noindent where $\tilde{R}$ is proportional to the Gaussian curvature of the 2D transverse space. Relative to this coordinate system, the differential equations that must be satisfied to ensure the Einstein condition \eqref{EinsteinC} are very complicated.

\item {\bf Case II.2a}:

The line element relative to the coordinates $(u,v,x^3,x^4)$ is:
\beq \d s^2 = - \d u \left( dv - \frac{v^2}{2} (C e^{f} - g_{33} b^2 e^{2f}) \d u - b e^f \d x^3 - f_{,x^4} \d x^4 \right) + g_{pq}(x^4) \d x^{p} \d x^{q}, \nonumber \eeq
\noindent where $f$ and $g_{pq}$ are functions of $x^4$ and the indices $p, q \in [3,4]$.  Then relative to the type ${\bf D^k}$ null frame, $\{ \bn, \bl, \bm^3, \bm^4\}$ with $i, j \in [3,4]$ the Ricci tensor components are:
\beq \begin{aligned}  
R_{12} &= Ce^f - g_{33} e^{2f} - \frac12 [ \Gamma^{4~3}_{~3} m_{4}^{~p} w_{p} + m_{i}^{~q}(m^{i~p} w_{p})_{,q} + \Gamma^{k~i}_{~i} m_{k}^{~p} w_{p} ]  \\
&  \hspace{ 6 mm }  +  (m_3^{~3} e^{f} )^2 + \delta^{ij} m_{i}^{~q} m_{j}^{~p} w_{q} w_{p} \\
R_{33} &=  - \Gamma^{4}_{~33} m_{4}^{~p} w_{p} + \frac12 (m_3^{~3} e^f)^2 + \tilde{R} \\
R_{34} &= m_{4}^{~q} (m_3^{~3} e^f)_{,q} - \Gamma^k_{~(34)} m_k^{~p} w_p + \frac12 m_3^{~3} m_{4}^{~p} e^f w_{p}  \\ 
R_{44} &= m_{4}^{~q}( m_{4}^{~p} w_{p})_{,q} - \Gamma^k_{~44} m_k^{~p} w_{p} + \frac12 m_{4}^{~q} m_{4}^{~p} w_{q} w_{p} + \tilde{R}, \end{aligned} \eeq

\noindent where $w_p = f_{,x^p} + \delta_{3p} be^f$ and $\tilde{R}$ is proportional to the Gaussian curvature of the 2D transverse space. In this case, the differential equations arising from imposing the Einstein condition \eqref{EinsteinM} are not readily solved in this coordinate system.

\end{itemize}

\bibliographystyle{unsrt-phys}
\bibliography{DkReferences}

\end{document}